\documentclass[conference, letterpaper]{IEEEtran}
\IEEEoverridecommandlockouts
\usepackage {textcomp, gensymb}
\usepackage{amssymb}
\usepackage{amsmath}
\usepackage{bm}
\usepackage{comment}
\usepackage{amsthm}
\usepackage{graphicx}
\usepackage{gensymb}
\usepackage{subfigure}
\usepackage{float}
\usepackage{cite}
\usepackage{algorithm,algorithmicx}
\usepackage{algpseudocode}
\usepackage{setspace}
\usepackage{enumerate}
\usepackage{enumitem}
\usepackage{multirow}
\usepackage{amsfonts}
\usepackage{url}
\usepackage{diagbox}
\usepackage{color, soul}
\usepackage{amsthm}
\usepackage{type1cm}
\DeclareMathOperator*{\argmax}{arg\,max}

\algnewcommand\algorithmicinput{\textbf{INPUT: }}
\algnewcommand\Input{\item[\algorithmicinput]}
\algnewcommand\algorithmicoutput{\textbf{OUTPUT: }}
\algnewcommand\Output{\item[\algorithmicoutput]}

\newtheorem{theorem}{\bf~~Theorem}

\newtheorem{lemma}{\bf~~Lemma}
\usepackage[bottom=1.05in,top=0.71in,left=0.680in,right=0.680in]{geometry}
\begin{document}

\title{Channel Estimation for Holographic Communications in Hybrid Near-Far Field}

\author{
\vspace{-0.8 cm}\\
\IEEEauthorblockN{
\normalsize{Shaohua Yue}\IEEEauthorrefmark{1},
\normalsize{Shuhao Zeng}\IEEEauthorrefmark{1},
\normalsize{Liang Liu}\IEEEauthorrefmark{2},
\normalsize{and Boya Di}\IEEEauthorrefmark{1}\\}
\IEEEauthorblockA{
    \normalsize{\IEEEauthorrefmark{1}State Key Laboratory of Advanced Optical Communication Systems and Networks},\\
    \normalsize{School of Electronics, Peking University, Beijing, China}\\
\normalsize{\IEEEauthorrefmark{2}The Hong Kong Polytechnic University, Hong Kong SAR, China}\\
\normalsize{Email: yueshaohua@pku.edu.cn, shuhao.zeng@pku.edu.cn, liang-eie.liu@polyu.edu.hk, boya.di@pku.edu.cn}\\
} 
%\thanks{This work was supported in part by the National Key R\&D Project of China under Grant No. 2022YFB2902800, National Science Foundation under Grant 62271012 and 62227809, and Beijing Natural Science Foundation under Grant 4222005 and L212027.}
%\vspace{-0.8cm}
}

\maketitle

\begin{abstract}
To realize holographic communications, a potential technology for spectrum efficiency improvement in the future sixth-generation (6G) network, antenna arrays inlaid with numerous antenna elements will be deployed. 
However, the increase in antenna aperture size makes some users lie in the Fresnel region, leading to the hybrid near-field and far-field communication mode, where the conventional far-field channel estimation methods no longer work well. To tackle the above challenge, this paper considers channel estimation in a hybrid-field multipath environment, where each user and each scatterer can be in either the far-field or the near-field region. First, a joint angular-polar domain channel transform is designed to capture the hybrid-field channel's near-field and far-field features. 
We then analyze the \emph{power diffusion} effect in the hybrid-field channel, which indicates that the power corresponding to one near-field (far-field) path component of the multipath channel may spread to far-field (near-field) paths and causes estimation error. We design a novel power-diffusion-based orthogonal matching pursuit channel estimation algorithm (PD-OMP). It can eliminate the prior knowledge requirement of path numbers in the far field and near field, which is a must in other OMP-based channel estimation algorithms. Simulation results show that PD-OMP outperforms current hybrid-field channel estimation methods. 
\end{abstract}

\begin{IEEEkeywords}
Holographic communication, channel estimation, power diffusion, near-field communication.
\end{IEEEkeywords}

\section{Introduction}
%To support the development of emerging applications such as autonomous driving and smart city, we expect the future sixth-generation (6G) mobile network to achieve higher spectrum efficiency.
To achieve the high spectrum efficiency required by the future sixth-generation (6G) network, holographic communication is a promising solution, where numerous antenna elements are integrated into a compact two-dimensional surface~\cite{holocom}. 
%To fulfill the stringent requirement, a key conceptual enabler is holographic communications, where numerous tiny antennas or reconfigurable elements are integrated into a compact two-dimensional surface~\cite{holocom}. 
Potential implementation technologies include reconfigurable holographic surface~\cite{RHS} and extremely large reconfigurable intelligent surface~\cite{LIS}. 
% One of the critical features holographic communications adopt to improve spectrum efficiency is applying an antenna array inlaid with a massive amount of antenna elements. 
Due to the increased radiation aperture size of the antenna array in holographic communications, the Fresnel region (radiating near-field region of the antenna) is significantly enlarged~\cite{Polar Domain}. As a result, a part of the users and scatterers are in the near-field region of the antenna array~\cite{nfs}, where the electromagnetic (EM) waves are characterized by spherical waves, while the others are located in the far-field region and the EM waves are modeled via uniform plane waves. This gives rise to the so-called \emph{hybrid near-far field} communication~\cite{hybridce1}. 

%Different from the conventional MIMO communications where all users are in the far field, the hybrid field leads to the power diffusion effect, i.e.,  the power corresponding to one near-field (far-field) path component of the multipath channel can spread to far-field (near-field) paths.
% 讲讲power diffusion

Most existing works focus on either the near-field~\cite{Polar Domain} or the far-field channel estimation~\cite{classicomp}. In \cite{classicomp, Polar Domain}, the polar domain and angular domain channel representation are proposed, respectively, so that the channel estimation problem can be reformulated as a compressed sensing problem and then solved by the orthogonal match pursuit (OMP) algorithm. In \cite{nfce}, the near-field channel estimation problem considering the non-stationarity is investigated, where the near-field region is divided into grids to perform the on-grid estimation. 
Few initial works \cite{hybridce1,hybridce2} consider the concept of a hybrid-field channel. They design channel estimation methods relying on the prior knowledge of the numbers of near-field paths and far-field paths such that the near-field and far-field path components are estimated separately.  

However, the above existing works have not considered the \emph{power diffusion} effect in the hybrid field, i.e., the power corresponding to one near-field (far-field) path component of the multipath channel may spread to far-field (near-field) paths. Therefore, the channel sparsity~\cite{sparsity} used in the far-field or near-field channel estimation does not hold. Moreover, for the general case where the numbers of the near-field paths and the far-field paths are unknown, channel estimation methods proposed in \cite{hybridce1,hybridce2} are not applicable.
%For the general case where the numbers of... 不知道， [8] [9] 分别估计的方法不适用。段落结束these methods rely on channel sparsity~\cite{sparsity} which does not hold for the hybrid-field channel~\cite{hybridce1} due to the power diffusion effect. Little research has been done to address the hybrid-field channel estimation. In~\cite{hybridce2}, the authors assume that the numbers of near-field and far-field paths are priorly known, such that the numbers of paths in the near field and far field can be applied to estimate the far-field and near-field path components of the hybrid-field channel separately.  

%本文中，我们考虑了一个 general的情况，power diffusion 不需要numbers已知，进行信道估计。 new challenges 
In this paper, we consider the hybrid-field channel estimation where the power diffusion effect is considered, and no prior knowledge of the numbers of near-field paths and far-field paths is required. To this end, two challenges have arisen. \emph{First}, it is non-trivial to distinguish the far-field and the near-field paths in this case because the boundary of the near-field and far-field region changes with the path direction~\cite{channelmodel} and is hard to specify. \emph{Second}, the channel sparsity is damaged due to the power diffusion effect, which leads to the performance degradation of existing channel estimation algorithms. It is critical to consider the power diffusion effect to improve the estimation accuracy.  
 
To cope with the above challenges, we first propose the joint angular-polar domain channel representation, based on which the channel sparsity is partially reserved and the far-field paths can be distinguished from the near-field paths. The power diffusion effect in the hybrid field case is then analyzed, and a power diffusion-based OMP channel estimation algorithm (PD-OMP) is developed, which can overcome the performance degradation caused by the power diffusion effect and enhance the estimation accuracy without the prior knowledge of the numbers of paths in the far field and near field. Finally, the effectiveness of the proposed algorithm is proved through the simulation results.

\section{System Model}
\subsection{Scenario Description}
\begin{figure}[ht]
\centerline{\includegraphics[width=8cm,height=4.8cm]{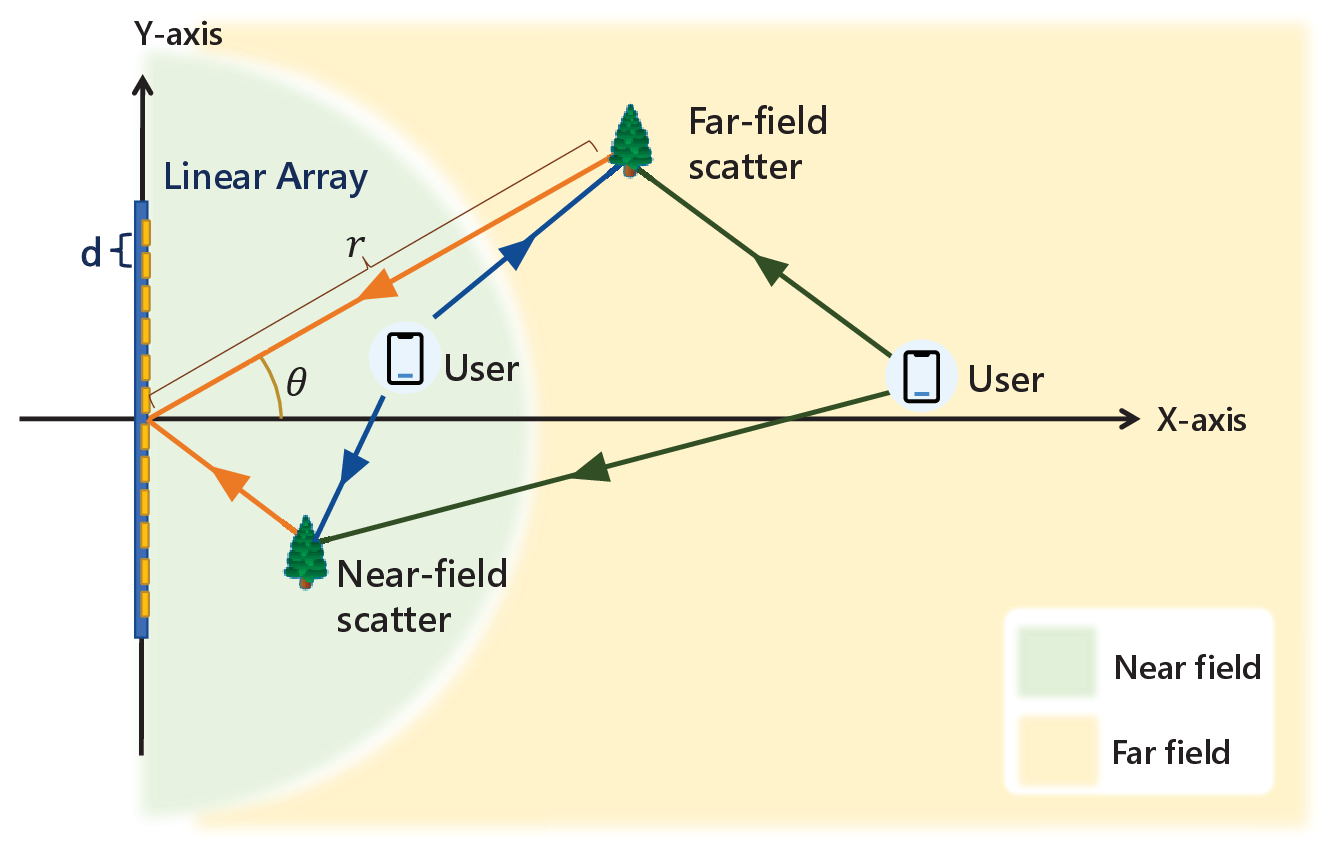}}
\vspace{-0.3cm}
%\label{systemmodel}
\caption{The hybrid channel model with a linear holographic antenna array.}
\vspace{-0.5cm}
\label{scenario}
\end{figure}

As shown in Fig.~\ref{scenario}, we consider the uplink communication in a multi-user holographic communication network, where the base station (BS) adopts an extremely large linear antenna array to communicate with multiple users, with the number of antenna elements and the element spacing denoted by $N$ and $d$, respectively. Moreover, we assume that the antenna elements are connected via $N_{RF}<N$ radio frequency (RF) chains such that the hybrid precoding scheme is employed at the BS.
The EM radiation field of the antenna array can be divided into the near field and far field, as indicated in Fig.~\ref{scenario}, and the boundary between these two fields depends on Rayleigh distance, which is positively correlated with the size of the antenna array \cite{channelmodel}. In our considered large antenna array system, Rayleigh distance is comparable to the cell radius~\cite{cellrange}. This thus leads to a \emph{hybrid-field} communication model, where the users and scatterers can locate in either the near field or the far field of the antenna array.
We assume that during the uplink channel estimation phase, communication resources are orthogonally assigned to different users for pilot signal transmission. Therefore, the channels for different users can be estimated independently. In the following, we only focus on the channel model and the channel estimation corresponding to \emph{one arbitrary user}.
\subsection{Hybrid-field Channel model}

Assume that the hybrid-field multipath channel from the considered user
%\footnote{The considered user locates in either the near field or the far field of the antenna array.} 
to the antenna array consists of $K$ non-line-of-sight (NLOS) paths\footnote{For simplicity, The LOS path is not discussed. Since the user can locate in either the near field or the far field, the LOS path can be discussed in a similar way as the NLOS path.}, where the user and each of the scatterers can lie in either the near field or the far field of the antenna array. In the following, we refer to a path as a far-field (near-field) path if the scatterer corresponding to the path is in the far field (near field) of the antenna array. Among the $K$ paths, the numbers of the far-field and near-field paths from the user to the antenna array are denoted by $K_F$ and $K_N=K-K_F$, respectively. We first present the model for the far-field and near-field paths, respectively, which are then combined to obtain the overall hybrid-field channel.
For simplicity, this paper considers a 2D Cartesian coordinate system, where the y-axis is aligned with the linear antenna array. Besides, the location of the middle point of the antenna array is set to be $(0,0)$, which shows that the x-axis is mid-perpendicular to the antenna array, as depicted in Fig.~\ref{scenario}.
%Note that in the hybrid-field case, a part of the scatterers and users are distributed within the near field of the antenna array, and the others are within the far field. The corresponding scattering paths and direct paths can be described with the far-field channel model~\cite{farfieldtransform} and the near-field channel model~\cite{channelmodel}, respectively. Hence, in this section, we will first present the far-field channel model and the near-field channel model and then combine both models to derive the hybrid-field channel model.

%where the $|\cdot|$ denotes the number of elements in a set. 

\subsubsection{Model for far-field paths}
For scatterers located in the far field of the antenna array, the EM wave received by the antenna array from the user can be approximated by the uniform plane wave. To embody this feature, the following model~\cite{channelmodel} can be utilized to describe the effect of the $k$-th far-field path on the transmitted EM signal, i.e.,
%\footnote{For convenience, the index $k = 0$ is used to identify the user.},
\begin{equation}
    \mathbf{h}^{F}_k = g_k \mathbf{a}(\theta_k),
\end{equation}
where $g_k$ is a complex factor describing the joint impact of the scattering and the channel fading. 
%The value of $g_k$ is related to the user's position. %\footnote{For the direct link from the antenna array to the user, $g_k = 1$.}} 
$\theta_k$ is the angle between the x-axis and the direction from the origin to the $k$-th scatterer, and $\mathbf{a}(\theta_k)$ represents the far-field steering vector towards $\theta_k$, i.e.,
\begin{equation}
    \mathbf{a}(\theta_k)= \frac{1}{\sqrt{N}}[1, e^{j\frac{2\pi d}{\lambda} \sin(\theta_k)}, ..., e^{j\pi \frac{2\pi(N-1)d}{\lambda} \sin(\theta_k)}]^{T}.
\end{equation}

\subsubsection{Model for near-field paths}
%相比于均匀球面波，平面波可以更好地描述……. \cite该球面波特性可以被下述模型捕捉(capture)到\cite
When considering scatterers located in the near field of the antenna array, the spherical wave model can describe the wavefront of EM waves more accurately compared with the plane wave. To capture this feature, the effect of the $k$-th near-field path on the transmitted EM signal is described as~\cite{nearfieldmodel}
\begin{equation}
        %\mathbf{h}^{N}_k = \frac{\lambda}{(4\pi)^{1.5}r_k}e^{-j\frac{2\pi}{\lambda}r_k} \mathbf{b}(\theta_k,r_k),
        \mathbf{h}^{N}_k = g_k \mathbf{b}(\theta_k,r_k),
\end{equation}
where $r_k$ is the distance between the $k$-th scatterer and the orgin and $\mathbf{b}(\theta_k,r_k)$ is the near-field steering vector. Here, it can be expressed as
\begin{equation}
    \mathbf{b}(\theta_k,r_k) = \frac{1}{\sqrt{N}}{[e^{-j\frac{2\pi}{\lambda}(r_{1,k}-r_k)},...,e^{-j\frac{2\pi}{\lambda}(r_{N,k}-r_k)}]}^T,
\end{equation}
where $r_{n,k}$ is the distance between the $n$-th antenna element of the antenna array and the $k$-th scatterer. $r_{n,k}$ can be expressed as
\begin{equation}
r_{n,k}=\sqrt{(r_k \cos\theta_k)^2+(t_n d-r_k\sin\theta_k)^2}.\label{distance}
\end{equation}
where $t_n = \frac{2n-N+1}{2}$ and $(0, t_nd)$ is the coordinate of the $n$-th antenna element.
\subsubsection{Overall Hybrid-field Channel Model}
 By combining $K_{N}$ near-field path components and $K_{F}$ far-field path components, the hybrid-field multipath channel is modeled as
\begin{equation}
    \mathbf{h} = \sum_{k = 1}^{K_{F}} \mathbf{h}_k^{F}+\sum_{k= K_F+1 }^{K_F+K_N}\mathbf{h}_k^{N}.
\end{equation}

\subsection{Signal Model}
 During uplink channel estimation, the user continuously transmits pilot symbols to the BS for $Q$ time slots\footnote{We assume that the channel coherence time is longer than the $Q$ time slots, so that the channel state information remains static during channel estimation.}. The received pilot $\mathbf{y}_q \in \mathbb{C}^{N_{RF}}$ of the BS at the $q$-th time slot is denoted as
\begin{equation}
    \mathbf{y}_q = \mathbf{W}_{q} \mathbf{h} x_q+\mathbf{W}_{q} \bm{n}_q,
\end{equation}
where $x_q$ is the transmitted pilot signal.
% and is set as $x_q=1$ for convenience. 
$\mathbf{W}_{q} \in \mathbb{C}^{N_{RF} \times {N}}$ is the hybrid beamforming matrix.
%, where $N_{RF}$ is the number of RF chains for the antenna array.
It is noted that $N_{RF} \ll N$ since the antenna array is featured with a great number of antenna elements. $\bm{n}_q \sim \mathcal{CN}(0, \sigma^{2} \mathbf{I}_{N \times 1})$ is the zero-mean complex Gaussian additive noise vector. Because no prior channel state information (CSI) is available in the channel estimation procedure, the phase shift for each element in the hybrid beamforming matrix is one-bit quantized and is randomly chosen with equal probability~\cite{Polar Domain}. 
The received pilot signal at the BS over the entire $Q$ time slots can be written as
\begin{equation}
    \mathbf{y} = \mathbf{W} \mathbf{h} x+\mathbf{W} \mathbf{n},
\end{equation}
where $\mathbf{y} = [\mathbf{y}_1^T, \mathbf{y}_2^T,... \mathbf{y}_Q^T]^T$,  
$\mathbf{W} = [\mathbf{W}_{1}^T, \mathbf{W}_{2}^T ..., \mathbf{W}_{ Q}^T]^T$, 
$\bm{n}=[\bm{n}_1^T, \bm{n}_2^T,...\bm{n}_Q^T]^T$ and $(\cdot)^T$ denotes the transpose operator.
\section{Channel Representation in Joint Angular-Polar Domain}
In existing works, channels are generally transformed to the angular domain~\cite{farfieldtransform} and the polar domain~\cite{Polar Domain} when all the scatterers and users are in the far field and the near field, respectively, so that sparse channel representations can be formulated. Based on such representations, channel estimation algorithms that reduce the pilot overhead are designed. However, in the hybrid-field case, neither the angular-domain transform nor the polar-domain transform is applicable because the sparsity cannot be guaranteed. Hence, we propose the joint angular-polar domain channel transform, based on which the sparsity is partially reserved and a low pilot overhead channel estimation method can be developed.
%在混合场情况下，单纯角域或者极化域下的信道矩阵的稀疏性无法得到保证。为此，我们引入了联合角域-极化域的概念来保证稀疏性。其主要思想是……。
%感觉有点奇怪，这个联合的角域和极化域是没法保证稀疏性吧？应该是联合的两个域+你的算法，能保证pilot低？
% To exploit the channel sparsity so that a channel estimation method with a low pilot overhead can be applied, the far-field channel and the near-field channel are transformed into the angular domain~ and the polar domain~, respectively. However, the sparsity of the hybrid-field channel does not hold in the angular domain and the polar domain. Hence, we propose the joint angular-polar domain for the hybrid-field channel to mitigate the sparsity issue. 

Note that the far-field channel is the weighted sum of steering vectors $\mathbf{a}(\theta)$ at the directions of the scatterers. A transform matrix $\mathbf{F}_{f}$ is designed to transform the channel $\mathbf{h}$ to its representation $\mathbf{h}^{a}$ in the angular domain, denoted as
\begin{equation}
    \mathbf{h} = \mathbf{F}_{f} \mathbf{h}^{a},\label{angular}
\end{equation}
where $\mathbf{F}_{f} = [\mathbf{a}(\theta_1), \mathbf{a}(\theta_2), ..., \mathbf{a}(\theta_N)]$ and $\theta_n = \arcsin{\frac{2n-1-N}{N}}, n=1,2...N$. 
%Because the number of multipath in the millimeter wave communication is usually limited \cite{Polar Domain}, a sparse far-field channel representation can be obtained in the angular domain. 
%The feature of channel sparsity can be exploited to reduce the complexity of the channel estimation algorithm.

%the near-field channel is also the weighted sum of steering vectors $\mathbf{b}(\theta, r)$ according to the directions and distances of the scatterers. 
Same as the angular domain, a transform matrix $\mathbf{F}_{n}$ comprised of near-field steering vectors is designed to transform the channel $\mathbf{h}$ to its representation in the polar domain \cite{Polar Domain}, which is denoted as
\begin{equation}
    \mathbf{h} = \mathbf{F}_{n} \mathbf{h}^{p}.\label{polar}
\end{equation} 
$\mathbf{F}_{n}$ is obtained by sampling both angles and distances in the space. Specifically,  $\mathbf{F}_{n} = [\mathbf{F}_{n,1}, \mathbf{F}_{n,2}... \mathbf{F}_{n,S}]$, where $\mathbf{F}_{n,s} = [\mathbf{b}(\theta_1, r_{s,1}), \mathbf{b}(\theta_2, r_{s,2})... \mathbf{b}(\theta_N, r_{s,N})], s = 1,2... S$. The design of $\{\theta_n, r_{s,n}\}$ 
can be found in \cite{Polar Domain}. 
%However, unlike the angular-domain representation for the far-field channel component, the near-field channel sparsity cannot be fully ensured in the polar domain. This effect damages the performance of the channel estimation methods relying on channel sparsity, which be illustrated in Section \ref{powerdiffusionsec}.

Considering both the far-field path components and the near-field path components in the hybrid-field channel, its representation $\mathbf{h}^{j}$ in the joint angular-polar domain is denoted as
\begin{equation}
\mathbf{h} = \mathbf{F}_{j}\mathbf{h}^{j},\label{joint}
\end{equation}
where the transform matrix $\mathbf{F}_{j}$ is the combination of the angular-domain transform matrix $\mathbf{F}_f$ and the polar-domain transform matrix $\mathbf{F}_n$ and is defined as
\begin{equation}
    \mathbf{F}_{j} = [\mathbf{F}_{f}, \mathbf{F}_{n,1}, \mathbf{F}_{n,2}... \mathbf{F}_{n,S}].
    \label{transmatrix}
\end{equation}

%\section{Problem Formulation}
%We aim to obtain the CSI from the received pilot signal. However, due to the great number of antenna elements of the antenna array, i.e., $N > N_{RF}Q$, the estimation error is unacceptable if we apply the conventional MMSE channel estimation method in the case of a low signal-to-noise (SNR) ratio~\cite{Polar Domain}. Thanks to the joint angular-polar domain, the channel estimation problem is formulated into a sparse signal recovery problem to minimize the estimation error, which is denoted as
%\begin{align}
%\min_{\mathbf{h}^{H,j}} &\Vert \mathbf{y} - \mathbf{W}\mathbf{F}_{j}\mathbf{h}^{j} \Vert_2,\label{opti}\\
%\text{s.t. }
%& \Vert \mathbf{h}^{j} \Vert_0 = K, \label{sparse} 
%\end{align}
%where constraint (\ref{sparse}) ensures that the number of estimated channel components is equal to the number of multipath. One of the low-complexity compressed sensing algorithms to tackle such a signal recovery problem is the OMP~\cite{hybridce1}. However, owing to the power diffusion effect, the channel sparsity in the joint angular-polar domain gets worse, which gives rise to an estimation error for the OMP algorithm. Thus, the OMP algorithm should be modified to suit the hybrid-field channel estimation case. 

\section{Power Diffusion in the Joint Angular-Polar Domain} \label{powerdiffusionsec}
Power diffusion is the phenomenon that in the channel transformation result, the power corresponding to one near-field (far-field) path component may spread to other far-field (near-field) paths. To demonstrate power diffusion, we consider the joint-angular-polar-domain representation of a channel consisting of a far-field path and a near-field path, as shown in Fig.~\ref{powerdiffusion}. For simplicity, the near-field space is sampled with one distance so that the transform matrix is denoted as $\mathbf{F}_{j} = [\mathbf{F}_{f}, \mathbf{F}_{n,1}]$. The left part and the right part of the figure are the transform result based on $\mathbf{F}_{f}$ and $\mathbf{F}_{n,1}$, respectively. As shown in Fig.~\ref{powerdiffusion}, the power of the near-field path is not only concentrated in one steering vector in $\mathbf{F}_{n,1}$ but also spreads across multiple steering vectors in $\mathbf{F}_{f}$. Similarly, one steering vector in $\mathbf{F}_{f}$ and multiple steering vectors in $\mathbf{F}_{n,1}$  should be jointly applied to describe the far-field path. This power diffusion effect indicates that the sparsity of the hybrid-field channel does not hold in either the angular domain or the polar domain. However, the peak values corresponding to both paths still exist in the overall transform result, indicating that the channel sparsity is partially preserved in the joint angular-polar domain.
 %Researches in optics have studied this effect with the Fresnel diffraction patterns \cite{opticsource} and the power diffusion of the near-field channel in the angular domain is reported in \cite{Polar Domain}. 
  
\begin{figure}[t]
\centerline{\includegraphics[width=8.6cm,height=4.5cm]{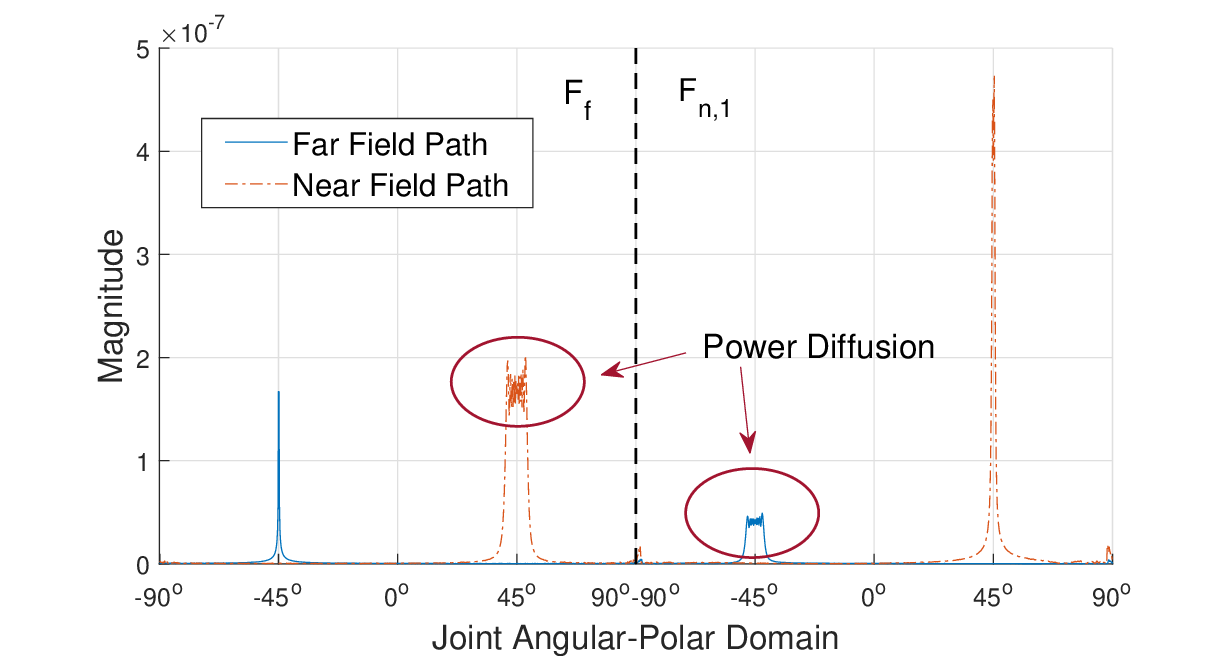}}

\caption{Power diffusion effect in the joint angular-polar domain.}
\vspace{-0.6cm}
\label{powerdiffusion}
\end{figure}

The reason for the power diffusion effect is that the coherence of two different steering vectors is not always zero, which is denoted by the following lemma.
\begin{lemma}
$\exists \{\theta_p,r_p\} \neq \{\theta_q,r_q\}, 
 \text{s.t.}$ $ \rm coherence$ $|b(\theta_p,r_p)^Hb(\theta_q,r_q)| \neq 0, $ $\rm coherence$ $|b(\theta_p,r_p)^Ha(\theta_q)| \neq 0,$ $\rm coherence$ $|a(\theta_p)^Hb(\theta_q, r_q)| \neq 0$
where $|\cdot|$ denotes the absolute operator. 
\end{lemma}
\begin{proof}
See Appendix~\ref{proofL1}.
\end{proof}

\vspace{-0.3cm}
 In the expression of coherence $|b(\theta_p,r_p)^Hb(\theta_q,r_q)|$ or $|b(\theta_p,r_p)^Ha(\theta_q)|$ or $|a(\theta_p)^Hb(\theta_q,r_q)|$, $b(\theta_p,r_p)$ or $a(\theta_p)$ represents a steering vector in the transform matrix and $b(\theta_q,r_q)$ or $a(\theta_q)$ represents a near-field or far-field path component. The transform of the channel to the joint angular-polar domain is performed with such computing of coherence. Due to the nonzero of coherence, steering vectors $b(\theta_p,r_p)$ or $a(\theta_p), \theta_p \neq \theta_q$  besides $a(\theta_q)$ are used to describe the channel component from the direction $\theta_q$, which matches the power diffusion effect in Fig.~\ref{powerdiffusion}. An approximation of the coherence calculation is given below. 
\begin{theorem}
    The coherence $\mu$ of two steering vectors $\{\mathbf{b}(\theta_p, r_p), \mathbf{b}(\theta_q,r_q)\}$ or $\{\mathbf{b}(\theta_p, r_p), \mathbf{a}(\theta_q)\}$ can be approximated as
    \begin{equation}
        \mu \approx \mathcal{F}(\beta, \rho) = |\int_{\frac{1-N}{2N}-\frac{\beta}{2N\rho d}}^{\frac{N-1}{2N}-\frac{\beta}{2N\rho d}}e^{-jkd^2\rho (Nx)^2}dx|,
    \label{pd}\end{equation}
 where $\beta = \sin\theta_p - \sin\theta_q$. $\rho = \frac{1-\sin^2\theta_p}{2r_p} - \frac{1-\sin^2\theta_q}{2r_q}$ for the case of $\{\mathbf{b}(\theta_p, r_p), \mathbf{b}(\theta_q,r_q)\}$, $\rho = \frac{1-\sin^2\theta_p}{2r_p}$ for the case of $\{\mathbf{b}(\theta_p, r_p)$ and $\rho = -\frac{1-\sin^2\theta_q}{2r_q}$ for the case of $\{\mathbf{a}(\theta_p), \mathbf{b}(\theta_q, r_q)\}$, $ \rho \neq 0$.
\end{theorem}
\begin{proof}
See Appendix~\ref{proofP1}.
\end{proof}

Due to the power diffusion effect, the performance of OMP methods based on solely angular-domain or polar-domain channel transform is deteriorated. The basic idea of OMP is to search for the peak values in the transform result. If only the angular-domain channel transform is considered, which is illustrated in the left part of Fig.~\ref{powerdiffusion}, multiple values can be falsely detected as the peak for the near-field path because of the power diffusion effect. The same issue exists for the polar-domain channel transform. 
However, in the joint angular-polar domain, peak values exist for both paths of the hybrid-field channel, which are larger than the magnitude of the path's power diffusion. If the steering vectors of the two peaks are successfully detected, the range of power diffusion can also be computed with (\ref{pd}). The detected path and its corresponding power diffusion can provide accurate CSI.
%The coherence approximation can be used to model the power diffusion effect. For each of the sampled angle and distance $\{\theta_q, r_q\}$ of the matrix $\mathbf{F}_{n,s}$, $\frac{1-\sin^2\theta_q}{2r_q}$ is a constant value \cite{Polar Domain}. After detecting one path in the polar domain, with the estimated $\theta_p$ and $r_p$, $\rho$ is determined, the diffusion of that path in $\mathbf{h}_{near,j}, \forall I \neq j$ can be computed based on the function $\mathcal{F}(\beta)$. 

\section{Power Diffusion-based Compressed Sensing Algorithm Design}
Based on the classic far-field OMP channel estimation algorithm \cite{classicomp}, a new hybrid-field OMP channel estimation algorithm considering the aforementioned power diffusion effect is proposed to improve the estimation accuracy. The proposed power diffusion-based OMP (PD-OMP) channel estimation algorithm is given in \textbf{Algorithm 1}. The main idea of PD-OMP 
is to perform the following three procedures in iteration: (1) searching for the steering vector that is most correlated with the residual pilot signal; (2) calculating the range of power diffusion; (3) eliminating the effect of the detected steering vectors in the residual signal.

In the beginning, to capture both the far-field and near-field features of the hybrid-field multipath channel, we generate the transform matrix $\mathbf{F}_j$ for the joint angular-polar domain according to its definition in (\ref{transmatrix}).
In step 2, we initialize the residual signal $\mathbf{R}$ as the received pilot signal $\mathbf{y}$ and the support set $\Gamma$ as the empty set. In step 3, the measurement matrix is set as $\mathbf{\Psi=WF}_{j}$.

\setlength{\textfloatsep}{0cm}
\vspace{0.1in}
\begin{algorithm}[htp]
%\footnotesize
%\footnotesize
\label{alg:MARL}
%\setstretch{0.95}

\caption{PD-OMP based Hybrid-field Channel Estimation} 
\hspace*{0.02in} {\bf Input:} Received pilot signal $\mathbf{y}$, power diffraction threshold $\alpha$, number of sampled distances $S$ of $\mathbf{F}_n$, number of paths $K$, the beamforming matrix $\mathbf{W}$. %\newline
%\hspace*{0.02in} {\bf Parameter:} maximum iteration times \emph{MaxIter}, step length \emph{MaxStep}.\newline
\begin{algorithmic}[1]
\State \textbf {Generate} the hybrid-field transform matrix $\mathbf{F}_{j}$ with $S$ based on (\ref{transmatrix}).
\State \textbf {Initialize} the support set $\Gamma = \{ \emptyset \}$ and the residue $\mathbf{R} = \mathbf{y}$.
\State \textbf{Set} the equivalent measurement matrix as $\mathbf{\Psi=WF}_{j}$.
\For{$i$ $= 1, 2, ..., K$}
    \State \textbf{Detect} a path $l^{*}$ based on (\ref{detect}) and obtain the corresponding direction and distance  $\{ \theta_{l^{*}}, r_{l^{*}}\}$. 
    %\Statex \qquad \quad based on (\ref{action}).
    \State \textbf{Compute} the power diffusion effect $\mathcal{F}$ with $\{ \theta_{l^{*}}, r_{l^{*}}\}$ based on (\ref{pd}).
    \State \textbf{Generate} the support set $\Gamma_{l^{*}}$ for path $l^{*}$ under the criterion of (\ref{criterion}).
    \State \textbf{Update} the support set $\Gamma$ = $\Gamma \cup \Gamma_{l^{*}}$.
    \State \textbf{Update} the residue as $\mathbf{R} = {\mathbf{y}} - \mathbf{\Psi}(:,\Gamma)\mathbf{\Psi}^{\dagger}(:,\Gamma){\mathbf{y}}$. % based on (\ref{softupdatea}), (\ref{softupdatec}).
\EndFor
\State \textbf{Compute} the estimated CSI $\hat{\mathbf{h}} = \mathbf{F}_{j}\mathbf{\Psi}^{\dagger}(:,\Gamma)\mathbf{y}$.
\end{algorithmic}\label{algo}
\hspace*{0.02in} {\bf Output:} 
The estimated CSI $\hat{\mathbf{h}}$.
\end{algorithm}

Then $K$ times of iteration are performed to find the steering vectors corresponding to each path component from the user to the antenna array. 
Specifically, in step 5, we first transform the residual signal to the joint angular-polar domain and then detect the strongest path $l^{*}$ as 
\begin{equation}
    l^{*} = \argmax_{l}|\mathbf{\Psi}(:,l)^H \mathbf{R}|^2, \label{detect}
\end{equation}
which indicates that the residual signal $\mathbf{R}$ has the strongest correlation with the $l^{*}$-th steering vector in $\mathbf{F}_j$. The corresponding direction and distance $\{ \theta_{l^{*}}, r_{l^{*}}\}$ associated with the $l^{*}$-th steering vector are obtained. In step 6, the power diffusion effect of the detected path can be computed with the integral approximation of coherence in Eq. (\ref{pd}). Specifically, the coherence of the $l^{*}$-th steering vector, i.e., the detected path,  and each steering vector in the transform matrix is computed. In step 7, the set $\Gamma_{l*}$ of the range of power diffusion is generated by the criterion as below,
\begin{equation}
    \Gamma_{l^{*}} = \{l | \mathcal{F}(\beta_l,\rho_l) \geq \alpha \},\label{criterion}
\end{equation}
where $\beta_l = \sin\theta_l-\sin\theta_{l^*}$,  $\rho_l = \frac{1-\sin^2\theta_l}{2r_l} - \frac{1-\sin^2\theta_{l^*}}{2r_{l^*}}$ and $\{\theta_l, r_l\}$ is associated with the $l$-th steering vector\footnote{If either the $l$-th or the $l^{*}$-th steering vector is a far-field steering vector, $r_l=\infty$ or $r_l^*=\infty$, respectively. If both the $l$-th and $l^{*}$-th steering vectors are far-field steering vectors, the $(l^{*}-1)$-th and $(l^{*}+1)$-th steering vectors are included in the power diffusion range~\cite{farfieldtransform} without computing $\mathcal{F}(\beta_l,\rho_l)$  as $\mathcal{F}(\beta_l,\rho_l)$ is incalculable for two far-field steering vectors.}. 
Criterion (\ref{criterion}) indicates that the coherence of each steering vector in the range of power diffusion and the detected $l^{*}$-th steering vector is no less than an adjustable parameter $\alpha$, which satisfies $0<\alpha\leq 1$. A small $\alpha$ indicates that a big range of power diffusion is considered. The effect of $\alpha$ on the performance of the PD-OMP will be discussed in the next section. The overall support set is updated with the union of the $\Gamma_{l^{*}}$ in step 8. The residue signal is updated by removing the projection of the detected paths in the received pilot signal with the least square method in step 9. Finally, the iteration is terminated and the hybrid field channel is recovered as 
%\begin{equation}
$\hat{\mathbf{h}} = \mathbf{F}_j\mathbf{\Psi}^{\dagger}(:,\Gamma)\mathbf{y}.$
%\end{equation}

\section{Simulation Results}
In this section, we evaluate the performance of the proposed channel estimation algorithm PD-OMP in terms of the normalized mean square error (NMSE). NMSE is defined as $\mathbb{E}\left\{\frac{\Vert \mathbf{\hat{h}} - \mathbf{h}\Vert_2^2}{\Vert \mathbf{h}\Vert_2^2}\right\}$, which represents the expectation of the relative estimation error.
In the simulation, we consider a holographic communication scenario where an antenna array at the working frequency of $30$ GHz is equipped with 200 antenna elements, and thus, the boundary between the near field and far field is approximately 200m~\cite{channelmodel}. The distance and angle of each user and scatterer to the origin satisfy the uniform distribution and are within the range of $(40m, 400m)$ and $(-60^{\circ}, 60^{\circ})$, respectively. 
%The magnitude of $g_k$ satisfies the Rayleigh distribution with the scale parameter $\sigma = 10^{-6}$. The phase of $g_k$ satisfies the uniform distribution within the range of $(0^{\circ},360^{\circ})$. 
$g_k$ satisfies circularly-symmetric complex Gaussian distribution.
Each element of the beamforming matrix $\mathbf{W}$ is randomly chosen from $\{\frac{1}{\sqrt{N}},-\frac{1}{\sqrt{N}}\} $ with equal probability. Other simulation parameters are listed in \textbf{Table 1}. To demonstrate the effectiveness of the proposed algorithm, we also compare it against the basic MMSE algorithm and four existing hybrid-field channel estimation algorithms, i.e.,
\begin{enumerate}    
    \item \emph{HF-OMP} \cite{hybridce1}: An OMP-based estimation method that requires the numbers of near-field paths and far-field paths as prior information and estimates the far-field and near-field path components separately. 
    \item \emph{HF-NPD OMP}: An OMP-based estimation method that only applies the joint angular-polar domain transformation (\ref{joint}) and does not consider the power diffusion effect.
    \item\emph{P-OMP} \cite{Polar Domain}: An OMP-based method that only applies the polar-domain transformation (\ref{polar}) and does not consider the power diffusion effect.
    \item\emph{A-OMP} \cite{classicomp}: An OMP-based method that only applies the angular-domain transformation (\ref{angular})  and does not consider the power diffusion effect.
    \item \emph{MMSE}: An estimation method applying the second-order statistics of the CSI to minimize the mean square error.
\end{enumerate}
\begin{table}[t]\footnotesize
  \centering
  \caption{Simulation Parameters} \label{}
  \begin{tabular}{cc}
     \hline
     Parameter & Value\\
     \hline   %  or \cline{col1-col2}
      Number of RF chain $N_{RF}$ & $10$ \\
      Number of users  & 5\\
      Number of paths $K$ & 5 \\
      Number of sampled distances $S$   & 4 \\
     \hline
  \end{tabular}
\end{table}

Fig.~\ref{snr} demonstrates the NMSE performance of different algorithms with the increase of SNR. Pilot length $Q = 10$ and $\alpha$ with minimum NMSE is selected for each SNR. 
Compared with the HF-OMP algorithm, which requires the numbers of near-field and far-field paths as prior knowledge, PD-OMP can estimate the channel more accurately without the prior knowledge of path distribution, which demonstrates the effectiveness of the proposed algorithm. PD-OMP also outperforms the A-OMP and P-OMP
methods, because PD-OMP applies the joint angular-polar domain transform matrix, which is capable of capturing both the far-field and near-field features of the channel. From Fig.~\ref{snr}, we can also find that the superiority of PD-OMP over the benchmark algorithms is more obvious when the SNR is closer to 20 dB.
This is because, at a high SNR, the range of power diffusion can be estimated more accurately, which is then utilized by our proposed algorithm to compensate for the performance degradation caused by the power diffusion effect. In contrast, none of the existing algorithms consider the power diffusion effect.   
%\begin{figure}[tbp]
%\centerline{\includegraphics[width=8.6cm,height=6cm]{nmsevssnr.eps}}
%\caption{The NMSE performance of different algorithms with the increase of SNR.}
% \vspace{-0.6cm}
%\label{snr}
%\end{figure}

\begin{figure*}[ht]
\begin{minipage}[t]{0.3\linewidth}
\centering

\vspace{-4mm}
\includegraphics[width=1\textwidth]{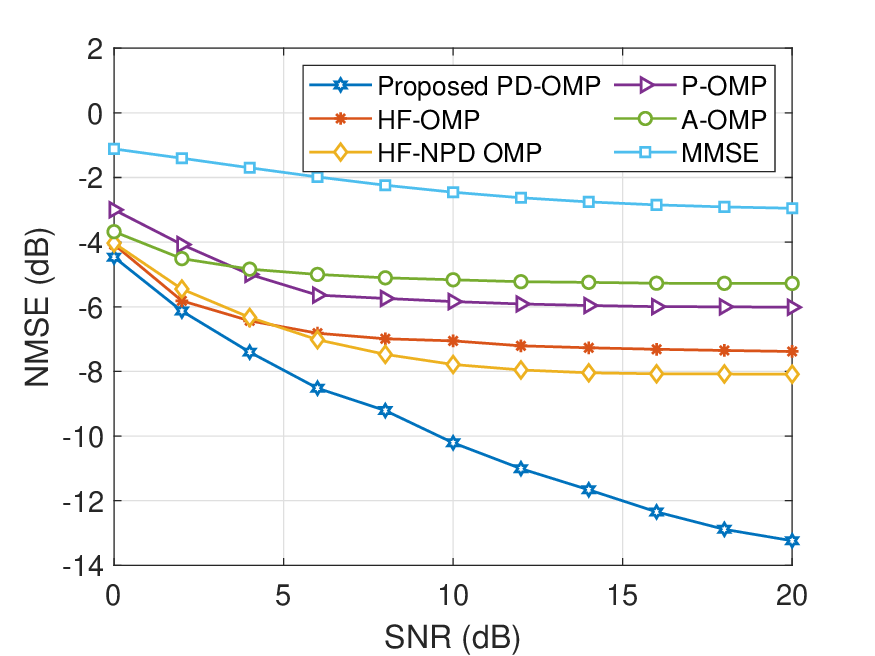}
\vspace{-8mm}
\caption{The NMSE performance of different algorithms with the increase of SNR. }
\label{snr}
\end{minipage}
\hfill
\begin{minipage}[t]{0.31\linewidth}
\centering

\vspace{-4mm}
\includegraphics[width=1\textwidth]{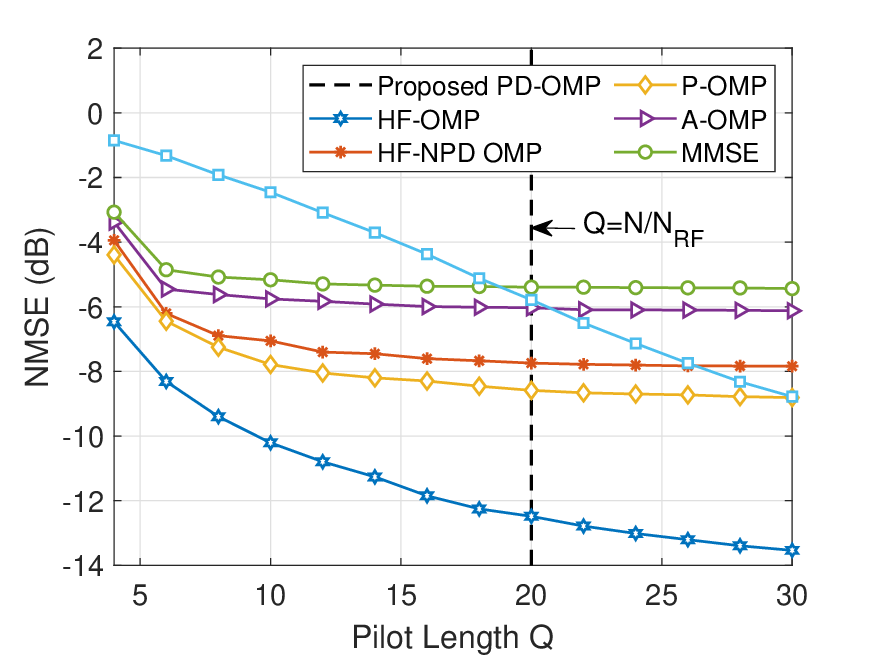}
\vspace{-8mm}
\caption{The NMSE performance of different algorithms with the increase of pilot length. }
\label{p}
\end{minipage}
\hfill
\begin{minipage}[t]{0.32\linewidth}
\centering

\vspace{-4mm}
\includegraphics[width=1\textwidth]{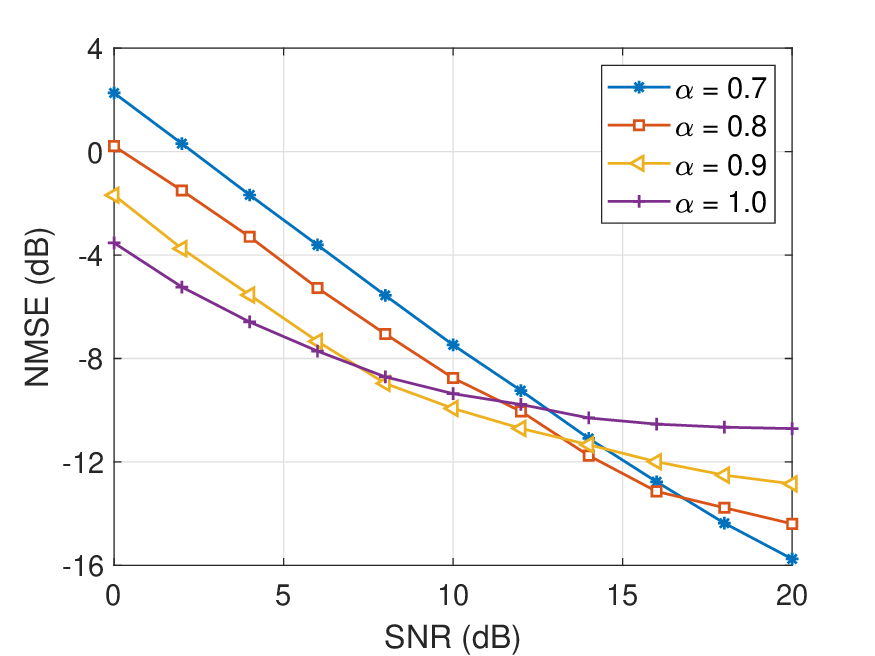}
\vspace{-8mm}
\caption{The NMSE performance of PD-OMP of different $\alpha$ with the increase of SNR.}
\label{alpha}
\end{minipage}
\vspace{-6mm}
\end{figure*}

Fig.~\ref{p} presents the NMSE performance of different algorithms versus pilot length $Q$. SNR $ = 10$ dB and $\alpha$ with minimum NMSE is selected for each $Q$. PD-OMP can achieve the lowest NMSE among all comparing hybrid-field channel estimation methods for different $Q$.  
%这里先说，不论pilot等于多少，你提出的性能总是最好证明了你的算法的有效性。之后说，OMP算法和MMSE算法性能分别随着pilot如何变化，之后解释原因。
The NMSE performance of all OMP-based algorithms first decreases and then tends to stabilize as $Q$ increases. This is because OMP is a compressed sensing algorithm able to recover high-dimension information from a low-dimension signal. On the contrary, the NMSE performance of MMSE keeps decreasing as $Q$ increases since MMSE requires the dimension of the received signal and estimated information to be similar to reach a low NMSE. Therefore, the PD-OMP method is preferred in the case where the pilot length is small, i.e., $Q < \frac{N}{N_{RF}}$.     

%\begin{figure}[tbp]
%\centerline{\includegraphics[width=8.6cm,height=6cm]{nmsevsp.eps}}
%\caption{The NMSE performance of different algorithms with the increase of pilot length.}
%\label{p}
%\end{figure}

As shown in Fig.~\ref{alpha}, we present how the NMSE performance of PD-OMP changes with SNR, where the influence of the power diffusion threshold $\alpha$ is studied.  $\alpha$ satisfies $0 < \alpha \leq 1$ and a smaller $\alpha$ indicates that a larger range of power diffusion is included for each estimated path. With the increase of SNR, PD-OMP with a decreased $\alpha$ can achieve the lowest NMSE. Note that the optimal NMSE is achieved only if a proper range of power diffusion is introduced in the algorithm. When SNR is high, the range of power diffusion is computed with high precision, and the estimation error brought by the power diffusion effect is eliminated more clearly with a smaller $\alpha$. Nevertheless, in the case of a low SNR, a falsely estimated range of power diffusion are likely to be introduced into the support set $\Gamma$, worsening the NMSE performance. Thus, a high $\alpha$ is advantageous when the SNR is low, as the range of power diffusion is limited. Fig.~\ref{alpha} also reveals that the SNR of the wireless communication scenario can be utilized to choose a proper range of power diffusion to improve the accuracy of channel estimation.

%\begin{figure}[t]
%\centerline{\includegraphics[width=8.6cm,height=6cm]{nmsevsalpha.eps}}
%\caption{The NMSE performance of PD-OMP of different $\alpha$ with the increase of SNR.}
% \vspace{-0.6cm}
%\label{alpha}
%\end{figure}
\vspace{-0.3cm}
\section{Conclusion}
\vspace{-0.3cm}
In this paper, we developed a channel estimation scheme for the hybrid-field multipath channel in holographic communications. Specifically, we first proposed the joint angular-polar domain channel transform, based on which we analyzed the power diffusion effect of the hybrid-field channel. Then we designed the channel estimation algorithm PD-OMP, which originated from OMP, and introduced the power diffusion to improve estimation accuracy. Simulation results showed that: 1) PD-OMP outperformed current state-of-the-art hybrid-field channel estimation methods under different SNRs. 2) PD-OMP achieved a $57.73\%$ estimation error reduction compared with MMSE when the ratio of the pilot length to the number of antenna elements is $0.04$. 3) The SNR of the holographic communication scenario could serve as useful information for setting the considered range of power diffusion in the algorithm to improve estimation accuracy.
% \section{acknowledgement}
\vspace{-0.3cm}
%\section{Acknowledgment}
%\vspace{-0.2cm}
%This work was supported in part by the National Key R\&D Project of China under Grant No. 2022YFB2902800, National Science Foundation under Grant 62271012 and 62227809, and Beijing Natural Science Foundation under Grant 4222005 and L212027.
%\vspace{-0.2cm}
\begin{appendices}
\section{proof of Lemma 1} \label{proofL1}

A example for proving the lemma is set as $N=200, \lambda=0.01m, \{\theta_p,r_p\}= \{ 40^{\circ},400m\}, \{\theta_q,r_q\}=\{45^{\circ},40m\}.$ The coherence $|b(\theta_p,r_p)^Hb(\theta_q,r_q)| = 0.0488$, $|b(\theta_p,r_p)^Ha(\theta_q)| = 0.0485$,
$|a(\theta_p)^Hb(\theta_q, r_q)| = 0.0489$.
Hence, the coherence is nonzero.
\vspace{-0.4cm}
\section{proof of Proposition 1} \label{proofP1}
In the case of computing the coherence of two near-field steering vectors, 
\begin{align}
    &\mu = |b(\theta_p,r_p)^Hb(\theta_q,r_q)| \nonumber \\ 
    &\overset{(a)}{\approx} \frac{1}{N}|\sum_{t_n} e^{j\frac{2\pi}{\lambda} t_n d (\sin\theta_p - \sin\theta_q) + j\frac{2\pi}{\lambda}t_n^2d^2(\frac{1-\sin^2\theta_p}{2r_p} - \frac{1-\sin^2\theta_q}{2r_q})}|\nonumber \\ 
    &= \frac{1}{N}|\sum_{n=\frac{1-N}{2}}^{\frac{N-1}{2}}e^{j\frac{2\pi}{\lambda}nd\beta-j\frac{2\pi}{\lambda}n^2d^2\rho}|\nonumber \\
    &= \frac{1}{N}|\sum_{n=\frac{1-N}{2}}^{\frac{N-1}{2}}e^{-j\frac{2\pi}{\lambda}d^2\rho((n-\frac{\beta}{2\rho d})^2-(\frac{\beta}{2\rho d})^2)}| \nonumber \\
    &\overset{(b)}{=} \frac{1}{N}|\sum_{n=\frac{1-N}{2}-\frac{\beta}{2\rho d}}^{\frac{N-1}{2}-\frac{\beta}{2\rho d}}e^{-j\frac{2\pi}{\lambda}d^2\rho n^2}| \nonumber\\
    &= |\sum_{n=\frac{1-N}{2N}-\frac{\beta}{2N\rho d}, \Delta n =\frac{1}{N}}^{\frac{N-1}{2N}-\frac{\beta}{2N\rho d}}e^{-jkd^2\rho (Nn)^2}\frac{1}{N}| \label{mc2} \\
     &\overset{N \rightarrow \infty}{\approx} |\int_{\frac{1-N}{2N}-\frac{\beta}{2N\rho d}}^{\frac{N-1}{2N}-\frac{\beta}{2N\rho d}}e^{-jkd^2\rho (Nx)^2}dx|=  \mathcal{F}(\beta, \rho).\nonumber
\end{align}
where the approximation $(a)$ is obtained by performing the second-order Taylor Expansion to the distance between the $p$-th scatterer and the $n$-th antenna element given in Eq. (\ref{distance}), which is denoted as
\begin{equation}
    r_{p,n} \approx r_p - t_n d \sin\theta_p + \frac{(1-\sin^2\theta_P)(t_n d)^2}{2 r_p}. \label{tx}
\end{equation}
Transform $(b)$ is derived by factoring out and cancelling the term $e^{jkd^2\rho(\frac{\beta}{2\rho d})^2}$ as it does not affect the modulus of $\mu$. Equation (\ref{mc2}) is further approximated by the integral form considering a large number of antenna elements of the antenna array.

For the cases of computing the coherence of $\{\mathbf{b}(\theta_p, r_p)$, $  \mathbf{a}(\theta_q)\}$ and $\{\mathbf{a}(\theta_p)$, $  \mathbf{b}(\theta_q, r_q)\}$, from Eq. (\ref{tx}) it can be seen that $\mathbf{a}(\theta) = \mathbf{b}(\theta, r)$ when $r \rightarrow \infty$. Therefore, $|\mathbf{b}(\theta_p,r_p)^H\mathbf{a}(\theta_q)|$ is transformed into $|\mathbf{b}(\theta_p,r_p)^H\mathbf{b}(\theta_q, \infty)|$ so the above proof for the coherence of two near-field steering vectors is also applicable, which ends the proof. 
\end{appendices}

\vspace{12pt}


\begin{thebibliography}{0}
%1
\bibitem{holocom}
C. Huang, \emph{et al.}, ``Holographic MIMO surfaces for 6G wireless networks: Opportunities, challenges, and trends," \emph{IEEE Wireless Commun.}, vol. 27, no. 5, pp. 118-125, Oct. 2020.

%2
\bibitem{RHS}
R. Deng, \emph{et al.}, ``Reconfigurable holographic surface: Holographic beamforming for metasurface-aided wireless communications," \emph{IEEE Trans. Veh. Technol.}, vol. 70, no. 6, pp. 6255-6259, Jun. 2021.

%3
\bibitem{LIS}
S. Zeng, \emph{et al.}, ``Reconfigurable intelligent surface (RIS) assisted wireless coverage extension: RIS orientation and location optimization," \emph{IEEE Commun. Lett.}, vol. 25, no. 1, pp. 269-273, Jan. 2021.

\bibitem{Polar Domain}
M. Cui and L. Dai, ``Channel estimation for extremely large-scale MIMO: Far-field or near-field?,” \emph{IEEE Trans. Commun.}, vol. 70, no. 4, pp. 2663-2677, Apr. 2022.


%5
\bibitem{nfs}
M. K. Ozdemir, H. Arslan and E. Arvas, ``On the correlation analysis of antennas in adaptive MIMO systems with 3-D multipath scattering," in \emph{IEEE Wireless Commun. Networking Conf. (WCNC)}, Atlanta, GA, USA, 2004, pp. 295-299.

%6
\bibitem{hybridce1}
X. Wei and L. Dai, ``Channel estimation for extremely large-scale massive MIMO: Far-field, near-field, or hybrid-field?," \emph{IEEE Commun. Lett.}, vol. 26, no. 1, pp. 177-181, Jan. 2022.

\bibitem{classicomp}
J. Lee, G. -T. Gil and Y. H. Lee, ``Channel estimation via orthogonal matching pursuit for hybrid MIMO systems in millimeter wave communications," \emph{IEEE Trans. Commun.}, vol. 64, no. 6, pp. 2370-2386, Jun. 2016.

%7
\bibitem{nfce}
Y. Han, S. Jin, C. -K. Wen and X. Ma, ``Channel estimation for extremely large-scale massive MIMO systems," \emph{IEEE Wireless Commun. Lett.}, vol. 9, no. 5, pp. 633-637, May 2020.

%9
\bibitem{hybridce2}
Z. Hu, \emph{et al.}, ``Hybrid-field channel estimation for extremely large-scale massive MIMO system," \emph{IEEE Commun. Lett.}, vol. 27, no. 1, pp. 303-307, Jan. 2023.

\bibitem{sparsity}
R. He,  \emph{et al.}, ``Wireless channel sparsity: Measurement, analysis, and exploitation in estimation," \emph{IEEE Wireless Commun.}, vol. 28, no. 4, pp. 113-119, Aug. 2021.



%10



%8


\bibitem{channelmodel}
H. Lu and Y. Zeng, ``Communicating with extremely large-scale array/surface: Unified modeling and performance analysis," \emph{IEEE Trans. Wireless Commun.}, vol. 21, no. 6, pp. 4039-4053, Jun. 2022.

%4
\bibitem{cellrange}
S. Sun, T. S. Rappaport and M. Shaft, ``Hybrid beamforming for 5G millimeter-wave multi-cell networks," \emph{IEEE Conf. Comput. Commun. Workshops (INFOCOM WKSHPS)}, Honolulu, HI, USA, 2018, pp. 589-596.

\bibitem{nearfieldmodel}
Z. Dong and Y. Zeng, ``Near-field spatial correlation for extremely large-scale array communications," \emph{IEEE Commun. Lett.}, vol. 26, no. 7, pp. 1534-1538, Jul. 2022.

%\bibitem{fixpn}
%W. Yu, \emph{et al.}, ``Hybrid far- and near-field channel estimation for THz ultra-massive MIMO via fixed point networks," in \emph{IEEE Glob. Commun. Conf. (GLOBECOM)}, Rio de Janeiro, Brazil, 2022, pp. 5384-5389.

\bibitem{farfieldtransform}
X. Gao, \emph{et al.}, ``Reliable beamspace channel estimation for millimeter-wave massive MIMO systems with lens antenna array," \emph{IEEE Trans. Wireless Commun.}, vol. 16, no. 9, pp. 6010-6021, Sept. 2017.

%11
%\bibitem{mmse}
%John G. Proakis, Masoud Salehi. ``Digital Communications.” New York, NY: McGraw Hill Press, 2008.

%\bibitem{l1}
%Eldar, Yonina and Kutyniok, Gitta. ``Compressed Sensing: Theory and Applications." Cambridge, UK: Cambridge University Press, 2012. 
%12

%\bibitem{opticsource}
%N. Decarli and D. Dardari, "Communication Modes With Large Intelligent Surfaces in the Near Field," \emph{IEEE Access}, vol. 9, pp. 165648-165666, Dec. 2021.

%13


%14 

\end{thebibliography}
\end{document}